\documentclass{article}

\usepackage{microtype}
\usepackage{fullpage}
\usepackage[utf8]{inputenc}
\usepackage{amsmath}
\usepackage{amssymb,amsthm}
\usepackage{sidecap}
\usepackage{tikz}
\usetikzlibrary{calc}
\usetikzlibrary{decorations.markings,decorations.pathmorphing}
\usetikzlibrary{shapes}
\usepackage{hyperref}
\usepackage{booktabs}
\usepackage{comment}
\usepackage{color}
\usepackage{boxedminipage}
\usepackage{caption,subcaption}
\usepackage{algorithm}
\usepackage{algpseudocode}
\usepackage{xspace}
\usepackage{wrapfig}
\usepackage{graphicx}
\usepackage{todonotes}

\DeclareMathOperator{\opt}{OPT}

\usepackage{xcolor}

\usepackage{mathtools}

\usepackage[all=normal,paragraphs=tight,bibliography=tight]{savetrees}

\usepackage{lmodern}
\usepackage[T1]{fontenc}

\usepackage{xspace}
\usepackage{tikz}
\usetikzlibrary{calc}
\usetikzlibrary{decorations.markings}

\newtheorem{theorem}{Theorem}[section]
\newtheorem{lemma}[theorem]{Lemma}
\newtheorem{proposition}[theorem]{Proposition}
\newtheorem{conjecture}{Conjecture}
\newtheorem{definition}[theorem]{Definition}
\newtheorem{remark}[theorem]{Remark}

\newcommand{\eps}{\epsilon}

\def\script#1{\mathcal{#1}}
\def\card#1{|#1|}
\def\set#1{\left\{#1\right\}}
\def\prob#1{\textup{\textsf{#1}}\xspace}

\def\etal{et al.\xspace}

\DeclareMathOperator{\polylog}{polylog}

\def\mypar#1{\smallskip\noindent \textbf{#1}}
\def\opt{\ensuremath{\mathrm{OPT}}}

\def\capacity{\mathrm{cap}}
\def\demand{\mathrm{dem}}

\def\Pairs{\script{M}}
\def\sT{\script{T}}

\def\bx{\mathbf{x}}
\def\bz{\mathbf{z}}
\def\sP{\script{P}} 

\def\prob#1{\textup{\textsc{#1}}\xspace}

\newcommand{\edp}{\prob{MaxEDP}}
\newcommand{\ndp}{\prob{MaxNDP}}

\newcommand{\termprs}{\mathcal{M}}
\newcommand{\Pfam}{\mathcal{P}}
\newcommand{\treedepth}{\mathrm{td}}
\newcommand{\Oh}{\mathcal{O}}

\newcommand{\bags}{\beta}
\newcommand{\adh}{\sigma}
\newcommand{\subtree}{\gamma}
\newcommand{\below}{\alpha}
\newcommand{\parent}{\mathrm{parent}}
\newcommand{\bd}{\delta}
\newcommand{\instance}{\mathcal{I}}
\newcommand{\capac}[1]{\mathrm{cap}(#1)}

\def\vd{\mathbf{d}}

\def\NP{\textsf{NP}}
\def\APX{\textsf{APX}}
\def\LP{\prob{\edp-LP}}

\begin{document}

\date{}
\title{On Routing Disjoint Paths in Bounded Treewidth Graphs}
%
%
\author{
Alina Ene\thanks{Department of Computer Science and DIMAP, University of Warwick. {\texttt{A.Ene@dcs.warwick.ac.uk}}}
\and
Matthias Mnich\thanks{Institut f{\"u}r Informatik, Universit{\"at} Bonn. {\texttt{mmnich@uni-bonn.de}}. Supported by ERC Starting Grant 306465 (BeyondWorstCase).}
\and
Marcin Pilipczuk\thanks{Institute of Informatics, University of Warsaw. \texttt{malcin@mimuw.edu.pl}. Research done while at University of Warwick, partially supported by DIMAP and by Warwick-QMUL Alliance in Advances in Discrete Mathematics and its Applications.}
\and
Andrej Risteski\thanks{Department of Computer Science, Princeton University. {\texttt{risteski@princeton.edu}}}
}

%
%
%

%
%

\maketitle              


\begin{abstract}
  We study the problem of routing on disjoint paths in bounded treewidth graphs with both edge and node capacities.
  The input consists of a capacitated graph $G$ and a collection of $k$ source-destination pairs $\Pairs = \set{(s_1, t_1), \dots, (s_k, t_k)}$.
  The goal is to maximize the number of pairs that can be routed subject to the capacities in the graph.
  A routing of a subset~$\Pairs'$ of the pairs is a collection $\mathcal{P}$ of paths such that, for each pair $(s_i, t_i) \in \Pairs'$, there is a path in $\mathcal{P}$ connecting $s_i$ to $t_i$.
  In the Maximum Edge Disjoint Paths (\edp) problem, the graph $G$ has capacities $\capacity(e)$ on the edges and a routing $\mathcal{P}$ is \emph{feasible} if each edge $e$ is in at most $\capacity(e)$ of the paths of $\mathcal{P}$.
  The Maximum Node Disjoint Paths (\ndp) problem is the node-capacitated counterpart of \edp.

  In this paper we obtain an $\Oh(r^3)$ approximation for \edp on graphs of treewidth at most~$r$ and a matching approximation for \ndp on graphs of pathwidth at most $r$.
  Our results build on and significantly improve the work by Chekuri \etal [ICALP 2013] who obtained an $\Oh(r \cdot 3^r)$ approximation for \edp.
   
\end{abstract}

\section{Introduction}
\label{sec:introduction}
In this paper, we study disjoint paths routing problems on bounded treewidth graphs. In this setting, we are given an undirected capacitated graph $G$ and a collection of source-destination pairs $\Pairs = \set{(s_1, t_1), (s_2, t_2), \dots, (s_k, t_k)}$. The goal is to select a maximum-sized subset $\Pairs' \subseteq \Pairs$ of the pairs that can be \emph{routed} subject to the capacities in the graph. More precisely, a routing of $\Pairs'$ is a collection $\sP$ of paths such that, for each pair $(s_i, t_i) \in \Pairs'$, there is a path in~$\sP$ connecting $s_i$ to $t_i$. In the Maximum Edge Disjoint Paths (\edp) problem, the graph $G$ has capacities $\capacity(e)$ on the edges and a routing $\sP$ is \emph{feasible} if each edge $e$ is in at most $\capacity(e)$ of the paths of $\sP$. The Maximum Node Disjoint Paths (\ndp) problem is the node-capacitated counterpart of \edp.

Disjoint paths problems are fundamental problems with a long history and significant connections to optimization and structural graph theory. The decision versions of \edp and \ndp ask whether all of the pairs can be routed subject to the capacities. Karp showed that, when the number of pairs is part of the input, the decision problems are \NP-complete (the node disjoint paths is part of Karp's original list of \NP-complete problems \cite{Karp72}). In undirected graphs, \edp and \ndp are solvable in polynomial time when the number of pairs is a fixed constant; this is a very deep result of Robertson and Seymour \cite{robertson1995graph} that builds on several fundamental results in structural graph theory from their graph minors project.

In this paper, we consider the optimization problems \edp and \ndp when the number of pairs are part of the input. These problems are \NP-hard and the main focus in this paper is on approximation algorithms for these problems in bounded treewidth graphs. Although they may appear to be quite specialized at first, \edp and \ndp on capacitated graphs of small treewidth capture a surprisingly rich class of problems; in fact, as shown by Garg, Vazirani, and Yannakakis \cite{garg1997primal}, these problems are quite interesting and general even on trees.

\edp and \ndp have received considerable attention, leading to several breakthroughs both in terms of approximation algorithms and hardness results. \edp is \APX-hard even in edge-capacitated trees \cite{garg1997primal}, whereas the decision problem is trivial on trees; thus some of the hardness of the problem stems from having to select a subset of the pairs to route. Moreover, by subdividing the edges, one can easily show that \ndp generalizes \edp in capacitated graphs. However, node capacities pose several additional technical challenges and extending the results for \edp to the \ndp setting is far from immediate even in restricted graph classes and our understanding of \ndp is more limited. 

In general graphs, the best approximation for \edp and \ndp is an $\Oh(\sqrt{n})$ approximation~\cite{CKS-sqrtn,KolliopoulosS04}, where $n$ is the number of nodes, whereas the best hardness for undirected graphs is only $\Omega((\log{n})^{1/2 - \eps})$~\cite{andrews2010inapproximability}. Bridging this gap is a fundamental open problem that seems quite challenging at the moment. There have been several breakthrough results on a relaxed version of these problems where congestion is allowed\footnote{A collection of paths has an \emph{edge} (resp. \emph{node}) \emph{congestion} of $c$ if each edge (resp. node) is in at most $c \cdot \capacity(e)$ (resp. $c \cdot \capacity(v)$) paths.}. This line of work has  culminated with a $\polylog(n)$ approximation with congestion $2$ for \edp \cite{chuzhoy2012polylogarithmic} and congestion $51$ for \ndp~\cite{ChekuriE13}. In addition to the routing results, this work has led to several significant insights into the structure of graphs with large treewidth and to several surprising applications \cite{chekuri2013large}.

Most of the results for routing on disjoint paths use a natural multi-commodity flow relaxation as a starting point. A well-known integrality gap instance due to Garg \etal \cite{garg1997primal} shows that this relaxation has an integrality gap of $\Omega(\sqrt{n})$, and this is the main obstacle for improving the $\Oh(\sqrt{n})$ approximation in general graphs. The integrality gap example is an instance on an $n \times n$ grid that exploits a topological obstruction in the plane that prevents a large integral routing (see Fig.~\ref{fig:lp}). Since an $n \times n$ grid has treewidth $\Theta(\sqrt{n})$, it suggests the following natural and tantalizing conjecture that was asked by Chekuri \etal\cite{chekuri2009note}.

\begin{conjecture}[\cite{chekuri2009note}]
\label{conj:tw}
	The integrality gap of the standard multi-commodity flow relaxation for \edp (and \ndp) is $\Theta(r)$ with congestion $1$, where $r$ is the treewidth of the graph.
\end{conjecture}

Recently, Chekuri, Naves, and Shepherd \cite{cns-tw} showed that \edp admits an $\Oh(r \cdot 3^r)$ approximation on graphs of treewidth at most $r$. This is the first approximation for the problem that is independent of~$n$ and $k$, and the first step towards resolving the conjecture. One of the main questions left open by the work of Chekuri \etal \cite{cns-tw} --- that was explicitly asked by them --- is whether the exponential dependency on the treewidth is necessary. In this paper, we address this question and we make a significant progress towards resolving Conjecture~\ref{conj:tw}. 

\begin{theorem}
\label{thm:edp}
  The integrality gap of the multi-commodity flow relaxation is $\Oh(r^3)$ for \edp in edge-capacitated undirected graphs of treewidth at most $r$. Moreover, there is a polynomial time algorithm that, given a tree decomposition of $G$ of width at most $r$ and a fractional solution to the relaxation of value \opt, it constructs an integral routing of size $\Omega(\opt / r^3)$.
\end{theorem}

As mentioned above, \ndp in node-capacitated graphs is more general than \edp and it poses several additional technical challenges. In this paper, we give an $\Oh(r^3)$ approximation for \ndp on graphs of pathwidth at most $r$ with arbitrary node capacities. This is the first result for \ndp that is independent of $n$ and it improves the $\Oh(r \log r \log n)$ approximation of Chekuri \etal \cite{chekuri2009note}.

\begin{theorem}
\label{thm:ndp}
	The integrality gap of the multi-commodity flow relaxation is $\Oh(r^3)$ for \ndp in node-capacitated undirected graphs of pathwidth at most $r$. Moreover, there is a polynomial time algorithm that, given a path decomposition of $G$ of width at most $r$ and a fractional solution to the relaxation of value \opt, it constructs an integral routing of size $\Omega(\opt / r^3)$.
\end{theorem}

The study of routing problems in bounded treewidth graphs is motivated not only by the goal of understanding the integrality gap of the multi-commodity flow relaxation but also by the broader goal of giving a more refined understanding of the approximability of routing problems in undirected graphs. Andrews \etal \cite{andrews2010inapproximability} have shown that \edp and \ndp in general graphs cannot be approximated within a factor better than $(\log{n})^{\Omega(1/c)}$ even if we allow a constant congestion $c \geq 1$. Thus in order to obtain constant factor approximations one needs to use additional structure. However, this seems challenging with our current techniques and there are only a handful of results in this direction.

One of the main obstacles for obtaining constant factor approximations for disjoint paths problems is that most approaches rely on a powerful pre-processing step that reduces an arbitrary instance of \edp/\ndp to a much more structured instance in which the terminals\footnote{The vertices participating in the pairs $\Pairs$ are called \emph{terminals}.} are \emph{well-linked}. This reduction is achieved using the well-linked decomposition technique of Chekuri, Khanna, and Shepherd \cite{cks-stoc05}, which necessarily leads to an $\Omega(\log{n})$ loss even in very special classes of graphs such as bounded treewidth graphs. Chekuri, Khanna, and Shepherd \cite{chekuri2009edge} showed that the well-linked decomposition framework can be bypassed in planar graphs, leading to a $\Oh(1)$ approximation for \edp with congestion $4$ (the congestion was later improved by S\'eguin-Charbonneau and Shepherd \cite{seguin2011maximum} from~$4$ to $2$). This result suggests that it may be possible to obtain constant factor approximations with constant congestion for much more general classes of graphs. In particular, Chekuri \etal \cite{cns-tw} conjecture that this is the case for the class of all minor-free graphs.

\begin{conjecture}[\cite{cns-tw}]
\label{conj:minor-free}
	Let $\script{G}$ be any proper minor-closed family of graphs. Then the integrality gap of the multi-commodity flow relaxation for \edp is at most a constant $c_{\script{G}}$ when congestion~$2$ is allowed.
\end{conjecture}

A natural approach is to attack Conjecture~\ref{conj:minor-free} using the structure theorem for minor-free graphs given by Robertson and Seymour \cite{robertson1986graph,robertson2003graph} that asserts
that every such graph admits a tree decomposition where the size of every adhesion (the intersection of neighboring bags) is bounded,
and after turning the adhesions into cliques, every bag induces a structurally simpler graph: one of bounded genus, with potentially a bounded number of apices and vortices.
Thus in some sense, in order to resolve Conjecture~\ref{conj:minor-free}, one needs to understand the base graph class (bounded genus graphs with apices and vortices) and how to tackle bounded width tree decompositions.

The recent work of Chekuri \etal \cite{cns-tw} has made a significant progress toward resolving Conjecture~\ref{conj:minor-free} by providing a toolbox for the latter issue, and the only ingredient that is still missing is an algorithm for planar and bounded genus graphs with a constant number of vortices (in the disjoint paths setting, apices are very easy to handle). However, one of the main drawbacks of their approach is that it leads to approximation guarantees that are \emph{exponential} in the treewidth.
Our work strengthens the approach of Chekuri \etal \cite{cns-tw} and it gives a much more graceful \emph{polynomial} dependence in the approximation ratio.

\begin{theorem}
\label{thm:ksums}
  Let $\script{G}$ be a minor-closed class of graphs such that the integrality gap of the multi-commodity flow relaxation is $\alpha$ with congestion $\beta$. Let $\script{G}_{\ell}$ be the class of graphs that admit a tree decomposition where, after turning all adhesions into cliques, each bag induces a graph from $\script{G}$, and each adhesion has size at most~$\ell$. Then the integrality gap of the relaxation for the class $\script{G}_{\ell}$ is $\Oh(\ell^3) \cdot \alpha$ with congestion~$\beta + 3$. 
\end{theorem}

We also revisit the well-linked decomposition framework of Chekuri \etal~\cite{cks-stoc05} and we ask whether the $\Omega(\log{n})$ loss is necessary for very structured graph classes. For bounded treewidth graphs, we give a well-linked decomposition framework that reduces an arbitrary instance of \edp to node-disjoint instances of \edp that are \emph{well-linked}. The loss in the approximation for our decomposition is only $\Oh(r^3)$, which improves the guarantee of $\Oh(\log{r} \log{n})$ from Chekuri \etal~\cite{chekuri2009note} when $r$ is much smaller than $n$.

It is straightforward to obtain the improved well-linked decomposition from our algorithm for \edp. Nevertheless, we believe it is beneficial to have such a well-linked decomposition, given that well-linked decompositions are one of the technical tools at the heart of the recent algorithms for routing on disjoint paths, integral concurrent flows \cite{chalermsook2012approximation}, and flow and cut sparsifiers \cite{chuzhoy2012vertex}. In particular, we hope that such a well-linked decomposition will have applications to finding flow and cut sparsifiers with Steiner nodes for bounded treewidth graphs. A sparsifier for a graph $G$ with~$k$ source-sink pairs is a significantly smaller graph~$H$ containing the terminals (and potentially other vertices, called Steiner nodes) that approximately preserves multi-commodity flows or cuts between the terminals. Such sparsifiers have been extensively studied and several results are known both in general graphs  and in bounded treewidth graphs (see Andoni et al.~\cite{AndoniGK14} and references therein).

A different question one could ask for problems in bounded treewidth graphs is whether additional computational power beyond polynomial-time running time can help with \edp or \ndp. 
It is a standard exercise to design an $n^{\Oh(r)}$-time dynamic programming algorithm (i.e., polynomial for every constant~$r$) for \ndp in uncapacitated graphs of treewidth $r$, while the aforementioned results on hardness of \edp in capacitated trees~\cite{garg1997primal}
rule out similar results for capacitated variants.
Between the world of having $r$ as part of the input, and having $r$ as a fixed constant, lies the world of \emph{parameterized complexity}, that asks for algorithms (called \emph{fixed-parameter algorithms})
with running time $f(r) \cdot n^c$, where $f$ is any computable function, and $c$ is a constant independent of the parameter. 
It is natural to ask whether allowing such running time can lead to better approximation algorithms.
As a first step towards resolving this question, we show a hardness for \ndp parameterized by \emph{treedepth}, a much more restrictive graph parameter than treewidth (cf.~\cite{sparsity}).
\begin{theorem}\label{thm:ndp-lb}
\ndp parameterized by the treedepth of the input graph is $W[1]$-hard, even with unit capacities.
\end{theorem}
Consequently, the existence of an \emph{exact} fixed-parameter algorithm is highly unlikely. 
We remark that our motivation for the choice of treedepth as a parameter stems also from the observation that a number of algorithms using the Sherali-Adams hierarchy to approximate a somewhat related problem
of~\prob{Nonuniform Sparsest Cut} in bounded treewidth graphs~\cite{sparsestcut-ChlamtacKR10,sparsestcut-GuptaTW13}
in fact implicitly uses a rounding scheme based on treedepth rather than treewidth.

\mypar{Paper organization.}
The rest of this paper is organized as follows. In Sect.~\ref{sec:prelim}, we formally define the problems and the multi-commodity flow relaxation. In Sect.~\ref{sec:edp}, we give our algorithm for \edp in bounded treewidth graphs and prove Theorems~\ref{thm:edp} and \ref{thm:ksums}. In Sect.~\ref{sec:ndp}, we extend our \edp algorithm to \ndp in bounded pathwidth graphs, and prove Theorem~\ref{thm:ndp}.
We give a well-linked decomposition for edge-capacitated graphs in Sect.~\ref{sec:wl-decomp}. 
Finally, we prove Theorem~\ref{thm:ndp-lb} in Sect.~\ref{sec:ndp-lb}.

\section{Preliminaries}
\label{sec:prelim}

\textbf{Tree and path decompositions.}
In this paper all tree decompositions are rooted; that is, a tree decomposition
of a graph $G$ is a pair $(\sT,\bags)$ where $\sT$ is a rooted tree and $\bags: V(\sT) \to 2^{V(G)}$ is a mapping
satisfying the following properties: (i) for every $e \in E(G)$, there exists a node $t \in V(\sT)$ with $e \subseteq \bags(t)$;
(ii) for every $v \in V(G)$ the set $\{t : v \in \bags(t)\}$ is nonempty and connected in $\sT$.

For a node $t \in V(\sT)$, we call the set $\bags(t)$ the \emph{bag} at node $t$, while for an edge $st \in E(\sT)$,
the set $\bags(t) \cap \bags(s)$ is called an \emph{adhesion}. 
For a non-root node $t \in V(\sT)$, by $\parent(t)$ we denote the parent of $t$, and by
$\adh(t) := \bags(t) \cap \bags(\parent(t))$ the adhesion on the edge to the parent of~$t$, called henceforth \emph{the parent adhesion};
for the root node $t_0 \in V(\sT)$ we put $\adh(t_0) = \emptyset$.
For two nodes $s,t \in V(\sT)$, we denote by $s \preceq t$ if $s$ is a descendant of $t$, and put
$\subtree(t) := \bigcup_{s \preceq t} \bags(s)$, $\below(t) := \subtree(t) \setminus \adh(t)$,
and $G(t) := G[\subtree(t)] \setminus E(G[\adh(t)])$.

\tikzset{vertex/.style={minimum size=1mm,circle,fill=black,draw, inner sep=0pt},
         decoration={markings,mark=at position .5 with {\arrow[black,thick]{stealth}}}}
%
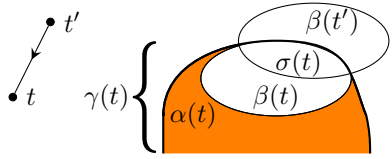
\begin{figure}
  \centering
\begin{tikzpicture}
  \node (t) at (-3.0,1.25) [vertex,label=right:$t'$] {};
  \node (s) at (-3.5,0.25)  [vertex,label=right:$t$] {};
  \draw[postaction={decorate}] (t)--(s);
  \draw[fill=orange,thick] (-1.5,-0.5) to [out=90,in=90](-1.5,0) to [out=90,in=180] (0,1) to [out=5,in=112](1,0.54) to [out=-65,in=112] (1.2,0) to [out=-78,in=90](1.25,-0.5);
  \draw[fill=white] (0.0,0.5) ellipse (1cm and 0.5cm);
  \draw[fill=white] (0.5,1.0) ellipse (1cm and 0.5cm);
  \draw[thick] (-1.5,0) to [out=90,in=180] (0,1) to [out=0,in=112](1,0.54) to [out=-65,in=112] (1.2,0);
  \draw (0.0,0.5) ellipse (1cm and 0.5cm);  
  \node at (0.75,1.25) {$\beta(t')$};
  \node at (0.0,0.25) {$\beta(t)$};
  \node at (0.3,0.7) {$\sigma(t)$};
  \node at (-1.1,0) {$\alpha(t)$};
  \node[yscale=1.7] at (-1.75,0.25) {{\Huge $\{$}};
  \node at (-2.25,0.25) {$\gamma(t)$};
\end{tikzpicture}
\caption{Notations used for a node $t$ with parent $t'$ in a tree decomposition $(\sT,\bags)$.}
\end{figure}

A \emph{torso} at node $t$ is a graph obtained from $G[\bags(t)]$ by turning every adhesion for an edge incident to $t$
into a clique.

We say that $(A, B)$ is a \emph{separation} in $G$ if $A \cup B = V(G)$ and there does not exist an edge of~$G$ with an endpoint in $A \setminus B$ and the other endpoint in $B \setminus A$.
We use the following well-known property of a tree decomposition.

\begin{lemma}[Lemma~{12.3.1} in \cite{diestel-book}]
\label{lem:adhesion-cutset}
  Let $(\sT,\bags)$  be a tree decomposition for a graph $G$.
  Then for every $t \in V(\sT)$ the pair $(\subtree(t), V(G) \setminus \below(t))$ is a separation of~$G$;
  note that $\subtree(t) \cap (V(G) \setminus \below(t)) = \adh(t)$.
\end{lemma}

A \emph{path decomposition} is a tree decomposition where $\sT$ is a path, rooted at one of its endpoints.

The width of a tree or path decomposition $(\sT, \bags)$ is defined as $\max_t |\bags(t)| - 1$. 
To ease the notation, we will always consider decompositions of width \emph{less} than~$r$, for some integer $r$,
so that every bag is of size at most $r$.

\medskip
\noindent
\textbf{Problem definitions.}
The input to \edp is an undirected graph $G$ with edge capacities $\capacity(e) \in \mathbb{Z}_+$ and a collection $\Pairs = \set{(s_1, t_1), \dots, (s_k, t_k)}$ of vertex pairs. A \emph{routing} for a subset $\Pairs'\subseteq\Pairs$ is a collection $\mathcal{P}$ of paths in $G$ such that, for each pair $(s_i, t_i) \in \Pairs'$, $\mathcal{P}$ contains a path connecting~$s_i$ to $t_i$. The routing is \emph{feasible} if every edge $e$ is in at most $\capacity(e)$ paths. In the Maximum Edge Disjoint Paths problem (\edp), the goal is to maximize the number of pairs that can be feasibly routed. The Maximum Node Disjoint Paths problem (\ndp) is the node-capacitated variant of \edp in which each node $v$ has a capacity $\capacity(v)$ and in a feasible routing each node appears in at most $\capacity(v)$ paths.

We refer to the vertices participating in the pairs $\Pairs$ as \emph{terminals}. It is convenient to assume that~$\Pairs$ form a matching on the terminals; this can be ensured by making several copies of a terminal and attaching them as leaves.

\medskip
\noindent
\textbf{Multicommodity flow relaxation.}
We use the following standard multicommodity flow relaxation for \edp (there is an analogous relaxation for \ndp). We use $\sP(u, v)$ to denote the set of all paths in $G$ from $u$ to $v$, for each pair $(u, v)$ of nodes. Since the pairs $\Pairs$ form a matching, the sets $\sP(s_i, t_i)$ are pairwise disjoint. Let $\sP = \bigcup_{i = 1}^k \sP(s_i, t_i)$. The LP has a variable $f(p)$ for each path $p \in \sP$ representing the amount of flow on $p$. For each pair $(s_i, t_i) \in \Pairs$, the LP has a variable $x_i$ denoting the total amount of flow routed for the pair (in the corresponding IP, $x_i$ denotes whether the pair is routed or not). The LP imposes the constraint that there is a flow from $s_i$ to $t_i$ of value $x_i$. Additionally, the LP has capacity constraints that ensure that the total amount of flow on paths using a given edge (resp. node for \ndp) is at the capacity of the edge (resp. node). 

\begin{figure}
\vspace{-0.2in}
\begin{center}
\begin{boxedminipage}{0.42\textwidth}
\vspace{-0.1in}
\begin{align*}
	& (\LP) &\\
	\quad \max \quad & \sum_{i = 1}^k x_i &\\
	\text{s.t.} \quad & \sum_{p \in \sP(s_i, t_i)} f(p) = x_i \leq 1,& i=1,\hdots,k\\
	& \sum_{p:\; e \in p} f(p) \leq \capacity(e), & e \in E(G)\\\
	& f(p) \geq 0, & p \in \sP \enspace .
\end{align*}
\end{boxedminipage}
\hspace{0.1in}
\begin{boxedminipage}{0.40\textwidth}
\centering
\includegraphics[scale=0.65]{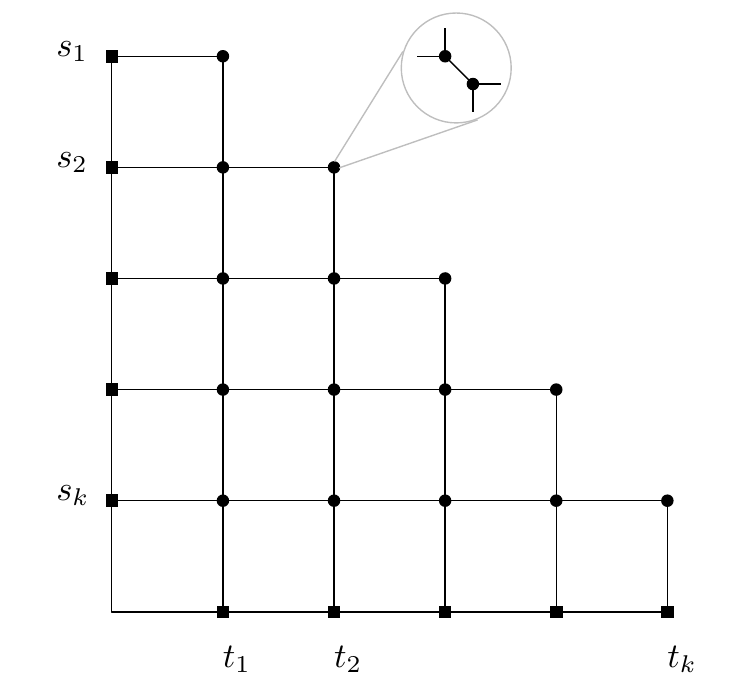}
\end{boxedminipage}
\vspace{-0.2in}
\end{center}
\caption{The multi-commodity flow relaxation for \edp. The instance on the right is the $\Omega(\sqrt{n})$ integrality gap example for \edp with unit edge capacities~\cite{garg1997primal}. Any integral routing routes at most one pair whereas there is a multi-commodity flow that sends $1/2$ units of flow for each pair $(s_i, t_i)$ along the canonical path from $s_i$ to $t_i$ in the grid.}
\label{fig:lp}
\end{figure}
It is well-known that the relaxation \LP can be solved in polynomial time, since there is an efficient separation oracle for the dual (alternatively, one can write a compact relaxation). We use $(f, \bx)$ to denote a feasible solution to \LP for an instance $(G, \Pairs)$ of \edp. For each terminal $v$, we also use $x(v)$ to denote the total amount of flow routed for $v$ and we refer to~$x(v)$ as the \emph{marginal value} of $v$ in the multi-commodity flow $f$.

\section{Algorithm for \edp in Bounded Treewidth Graphs}
\label{sec:edp}
We give a polynomial time algorithm for \edp that achieves an $O(r^3)$ approximation for graphs with treewidth less than $r$.
Our algorithm builds on the work of Chekuri \etal~\cite{cns-tw}, and it improves their approximation guarantee from $O(r \cdot 3^r)$ to $O(r^3)$.
We use the following routing argument as a building block.

\begin{proposition}[Proposition~{3.4} in \cite{cns-tw}]
\label{prop:small-cutset-routing}
	Let $(G, \Pairs)$ be an instance of \edp and let $(f, \bx)$ be a feasible fractional solution for the instance.
	If there is a second flow that routes at least $x(v)/\alpha$ units of flow for each $v$ to some set $S \subseteq V(G)$, where $\alpha \geq 1$ then there is an integral routing of at least ${\card{f} \over 36 \alpha \card{S}}$ pairs.
\end{proposition}

Our starting point is a tree decomposition $(\sT,\bags)$ for $G$ of width less than~$r$ and a fractional solution
$(f, \bx)$ to the multicommodity flow relaxation for \edp given in Section~\ref{sec:prelim},
that is, the flow $f$ routes $\bx(v)$ units of flow for each vertex $v \in V$.
We let $\card{f}$ denote the total amount of flow routed by $f$, i.e., $\card{f} = {1 \over 2} \sum_{v \in V} \bx(v)$.

The following definitions play a key role in our algorithm.

\begin{definition}[Safe node]
A node $t \in V(\sT)$ is \emph{safe} with respect to $(f, \bx)$
if there is a second multicommodity flow $g$ in $G(t)$ such that $g$ satisfies the edge capacities of $G(t)$ and,
for each vertex $z \in \subtree(t)$, $g$ routes at least ${1 \over 4r} \cdot \bx(z)$ units of flow from $z$ to the adhesion $\adh(t)$.
The node $t$ is \emph{unsafe} if it is not safe.
\end{definition}

\begin{definition}[Good node]
A node $t \in V(\sT)$ is \emph{good} with respect to $(f, \bx)$ if every flow path in the support of $f$ that
has an endpoint in $\subtree(t)$
also intersects~$\adh(t)$;
in other words, there does not exist a flow path that is completely contained in~$G[\below(t)]$.
A node is \emph{bad} if it is not good.
\end{definition}

\begin{remark}
If a node $t$ is good then it is also safe, as shown by the following multicommodity flow $g$ in~$G(t)$.
For each path $p$ in the support of $f$ that originates in $\subtree(t)$,
let $p'$ be the smallest prefix of $p$ that ends at a vertex of $\adh(t)$
(since $p$ intersects $\adh(t)$, there is such a prefix);
we set $g(p') = f(p)$. The resulting flow $g$ is a feasible multicommodity flow in $G(t)$ that routes $\bx(z)$ units of flow from $z$ to
  $\adh(t)$ for each vertex $z \in \subtree(t)$. Therefore, $t$ is safe.
\end{remark}

Our approach is an inductive argument based on the maximum size of a parent adhesion
that is bad or unsafe.
More precisely, we prove the following:

\begin{theorem} \label{thm:induction-main}
	Let $(G, \Pairs)$ be an instance of \edp and let $(f, \bx)$ be a fractional solution for $(G, \Pairs)$, where $f$ is a feasible multicommodity flow in~$G$ for $\Pairs$ with marginals $\bx$. Let $(\sT,\bags)$ be a tree decomposition for $G$ of width less than $r$.
Let $\ell_1$ be the maximum size of a parent adhesion of an unsafe node, and let $\ell_2$ be the maximum size of a
parent adhesion of a bad node.
There is a polynomial time algorithm that constructs an integral routing of size at least
${1 \over 144 r^3} \cdot \left(1 - {1 \over r} \right)^{\ell_1 + \ell_2} \cdot \card{f}.$
\end{theorem}

We start with a bit of preprocessing. If $\card{f} = 0$, then we return an empty routing.
Otherwise, the root node of $\sT$ is always unsafe and bad, and the integers~$\ell_1$ and $\ell_2$ are well-defined.
By considering every connected component of $G$ independently (with inherited tree decomposition from $(\sT,\bags)$),
we assume that~$G$ is connected; note that in this step all safe or good adhesions remain safe or good for every connected component.
Furthermore, we delete from $(\sT,\bags)$ all nodes with empty bags; note that the connectivity of $G$ ensures that
the nodes with non-empty bags induce a connected subtree of $\sT$.
In this step, the root of~$\sT$ may have moved to a different node (the topmost node
with non-empty bag), but the parent-children relation in the tree remains unchanged.

Once $G$ is connected and no bag is empty, 
the only empty parent adhesion is the one for the root node.
We prove Theorem~\ref{thm:induction-main} by induction on $\ell_1 + \ell_2 + \card{V(G)}$.

\mypar{Base case.}
In the base case, we assume that $\ell_1 = \ell_2 = 0$. Since every parent adhesion of a non-root node is non-empty,
that implies that the only bad node is the root $t_0$, that is, 
every flow path in $f$ passes through~$\bags(t_0)$, which is of size at most $r$. 
By applying Proposition~\ref{prop:small-cutset-routing} with $S = \beta(t_0)$ and $\alpha=1$,
we construct an integral routing of size at least $\frac{1}{36r} |f| \geq \frac{1}{144r^3}|f|$.

In the inductive step, we consider two cases, depending on whether $0 \leq \ell_1 < \ell_2$ or $0 < \ell_1 = \ell_2$.

\mypar{Inductive step when $0 \leq \ell_1 < \ell_2$.}
Let $\{t_1,t_2,\ldots,t_p\}$ be the topmost bad nodes of $\sT$ with parent adhesions of size $\ell_2$, that is, 
it is a minimal set of such bad nodes such that for every bad node $t$ with parent adhesion of size $\ell_2$,
there exists an $i\in\{1,\hdots,p\}$ with $t \preceq t_i$.
For $1 \leq i \leq p$, let $f^{\mathrm{inside}}_i$ be the subflow of $f$ consisting of all paths that are completely contained in
$G[\alpha(t_i)]$. Furthermore, since $\ell_1 < \ell_2$, the node $t_i$ is safe; let~$g_i$ be the corresponding flow, i.e., 
a flow that routes $\frac{1}{4r} x(v)$ from every $v \in \subtree(t_i)$ to $\adh(t_i)$ in $G(t_i)$.
By applying Proposition~\ref{prop:small-cutset-routing}, there is an integral routing~$\sP_i$ in $G(t_i)$
that routes at least ${1 \over 144 r^2} \card{f^{\mathrm{inside}}_i}$ pairs.
Since the subgraphs $\set{G(t_i) \colon 1 \leq i \leq p}$ are edge-disjoint,
we get an integral routing $\sP \coloneqq \bigcup_i \sP_{i}$ of size at least
${1 \over 144 r^2} \sum_{i=1}^p \card{f^\mathrm{inside}_i}$. 

If $\sum_{i=1}^p \card{f^{\mathrm{inside}}_i} > \frac{1}{r} \card{f}$, then we can return the routing $\sP$ as the desired
solution. Otherwise, we drop the flows $f^{\mathrm{inside}}_i$, that is, consider a flow
$f' \coloneqq f - \sum_{i=1}^p f^{\mathrm{inside}}_i$.
Clearly, $\card{f'} \geq (1-\frac{1}{r}) \card{f}$.
Furthermore, by definition of~$f^\mathrm{inside}_i$, every node $t_i$ is good with respect to $f'$.
Since deleting a flow path cannot turn a good node into a bad one nor a safe node into an unsafe one, and all descendants 
of a good node are also good, we infer that every unsafe node with respect to $f'$ has parent adhesion of size at most $\ell_1$,
while every bad node with respect to $f'$ has parent adhesion of size \emph{less} than $\ell_2$. 
Consequently, by induction hypothesis we obtain an integral routing of size at least
$$\frac{1}{144r^3} \left(1-\frac{1}{r}\right)^{\ell_1 + \ell_2 - 1} \card{f'} \geq \frac{1}{144r^3} \left(1-\frac{1}{r}\right)^{\ell_1 + \ell_2} \card{f}.$$

\mypar{Inductive step when $0 < \ell_1 = \ell_2$.}
In this case, we pick a node $t^\circ$ to be the lowest node of $\sT$ that 
is unsafe and has parent adhesion of size $\ell_1$.
By the definition of an unsafe node and Menger's theorem, there exists a set $U \subseteq \below(t^\circ)$
such that $\capac{\bd(U)} < \frac{1}{4r} \bx(U)$.
With a bit more care, we can extract a set $U$ with one more property:
\begin{lemma}\label{lem:belowU}
In polynomial time we can find a set $U \subseteq \below(t^\circ)$ for which
\begin{enumerate}
\item $\capac{\bd(U)} < \frac{1}{4r} \bx(U)$;\label{belowU:capac}
\item for every non-root node $t$, if $\adh(t) \subseteq U$, then $\subtree(t) \subseteq U$.\label{belowU:adh}
\end{enumerate}
\end{lemma}
\begin{proof}
Consider an auxiliary graph $G'$, obtained from $G[\subtree(t^\circ)]$ by adding a super-source $s^\ast$, linked for every $v \in \subtree(t^\circ)$ by an arc $(s^\ast,v)$ of capacity $\frac{1}{4r} \bx(v)$, and a super-sink $t^\ast$, linked
for every $v \in \adh(t^\circ)$ by an arc $(v,t^\ast)$ of infinite capacity. Let~$U$ be such a set that $\bd(U \cup \{s^\ast\})$
is a minimum $s^\ast$-$t^\ast$ cut in this graph.
Clearly, since~$U$ is unsafe,
$\capac{\bd_{G'}(U \cup \{s^\ast\})} < \frac{1}{4r} \bx(\subtree(t^\circ)) = \capac{\bd_{G'}(s^\ast)}$, so $U \neq \emptyset$.
Also, $U \subseteq \below(t^\circ)$, as each node in $\adh(t^\circ)$ is connected to $t^\ast$ with an infinite-capacity arc.

We claim that $U$ satisfies the desired properties. The first property is immediate:
$$\capac{\bd_G(U)} = \capac{\bd_{G'}(U \cup \{s^\ast\})} - \frac{1}{4r} \bx(\subtree(t^\circ) \setminus U) < \frac{1}{4r} \left(\bx(\subtree(t^\circ)) - \bx(\subtree(t^\circ) \setminus U) \right) = \frac{1}{4r} \bx(U).$$

For the second property, pick a non~root node $t$ with $\adh(t) \subseteq U$.
Since $\adh(t) \subseteq U \subseteq \below(t^\circ)$,
 we have $t \preceq t^\circ$, $t \neq t^\circ$, and $\subtree(t) \subseteq \below(t^\circ)$.
Let $U' \coloneqq U \cup \subtree(t)$. By Lemma~\ref{lem:adhesion-cutset},
$\bd_G(U') \subseteq \bd_G(U)$, and hence $\bd_{G'}(U' \cup \{s^\ast\}) \subseteq \bd_{G'}(U \cup \{s^\ast\})$.
However, since $\bd_{G'}(U \cup \{s^\ast\})$ is a minimum cut, we have actually $\bd_G(U') = \bd_G(U)$.
Since~$G$ is connected, this implies that $U = U'$, and thus $\subtree(t) \subseteq U$.
As the choice of~$t$ was arbitrary, $U$ satisfies the second property.
\end{proof}

Using the cut $U$, we split the graph $G$ and the flow $f$ into two pieces as follows.
Let $G_1 = G[U]$ and $G_2 = G-U$.
Let $f_i$ be the restriction of $f$ to $G_i$, i.e., the flow consisting of only flow paths that are contained in $G_i$.
Let $\bx_i$ be the marginals of $f_i$ and let $\Pairs_i$ be the subset of $\Pairs$ consisting of all pairs $(s, t)$ such that $\set{s, t} \subseteq V(G_i)$; note that $x_i(s) = x_i(t)$ for each pair $(s, t) \in \Pairs_i$ and thus $(f_i, \bx_i)$ is a fractional routing for the instance $(G_i, \Pairs_i)$. Let $(\sT,\bags_1)$ and $(\sT,\bags_2)$ be the restriction of $(\sT,\bags)$
to the vertices of $G_1$ and~$G_2$, respectively; we define mappings $\adh_i$, $\subtree_i$, and $\below_i$
naturally.
In what follows, we consider separately two instances
$\instance_i \coloneqq \left<(G_i, \Pairs_i), (f_i, \bx_i), (\sT,\beta_i)\right>$ for $i=1,2$.

An important observation is the following:
\begin{lemma}\label{lem:remains-good}
Every node $t \in V(\sT)$ that is good
in the original instance (i.e., as a node of $\sT$,
with respect to $(f,\bx)$) is also good in $\instance_i$ with respect to $(f_i,\bx_i)$.
\end{lemma}
\begin{proof}
Note that every flow path in $f_i$ is also present in $f$, and therefore intersects the parent adhesion of $f$ if $t$ is a good
node in the original instance.
\end{proof}
Consequently, every node $t \in V(\sT)$ with $|\adh(t)| > \ell_2$ is good in the instance $\instance_i$,
and the maximum size of a parent adhesion of a bad node in instance $\instance_i$ is at most $\ell_2$.
Hence, both $\instance_1$ and $\instance_2$ satisfy the assumptions of Theorem~\ref{thm:induction-main} with
not larger values of $\ell_1$ and $\ell_2$.
Furthermore, note that $|V(G_i)| < |V(G)|$ for $i=1,2$.

For $\instance_2$, the above reasoning allows us to
simply just apply inductive step,
obtaining an integral routing~$\sP_2$ of size at least

\begin{equation}\label{eq:inst2}
\card{\sP_2} \geq {1 \over 144 r^3} \left(1 - {1 \over r} \right)^{\ell_1 + \ell_2} \cdot \card{f_2} \enspace .
\end{equation}

For $\instance_1$, we are going to obtain a larger routing via an inductive step with better bounds.
\begin{lemma}\label{lem:inst1}
The size of the largest parent adhesion of an unsafe note in $\instance_1$ is \emph{less} than $\ell_1$.
\end{lemma}
\begin{proof}
Assume the contrary, let $t \in V(\sT)$ be an unsafe adhesion with $|\adh_1(t)| \geq \ell_1$.
If $|\adh(t)| > \ell_1$, then $t$ is good in the original instance, and by Lemma~\ref{lem:remains-good}
it remains good in $\instance_1$. Consequently, $|\adh(t)| = |\adh_1(t)| = \ell_1$; in particular,
$\adh(t) = \adh_1(t) \subseteq U$.

By Lemma~\ref{lem:belowU}, property~\ref{belowU:adh}, we have $\subtree(t) \subseteq U$.
Consequently, $t$ is safe in the original instance if and only if it is safe in $\instance_1$.
Since $t \preceq t^\circ$, $t \neq t^\circ$, but $|\adh(t)| = \ell_2$, by the choice of $t^\circ$
it holds that~$t$ is safe in the original instance, a contradiction.
\end{proof}
Lemma~\ref{lem:inst1} allows us to apply the inductive step to $\instance_1$ and obtain an integral routing $\sP_1$
of size at least

\begin{equation}\label{eq:inst1}
\card{\sP_1} \geq {1 \over 144 r^3} \left(1 - {1 \over r} \right)^{\ell_1-1 + \ell_2} \cdot \card{f_1} \enspace .
\end{equation}

Let us now estimate the amount of flow lost by the separation into $\instance_1$ and~$\instance_2$, i.e., 
$g = f - f_1 - f_2$. As every flow path in $g$ passes through $\bd(U)$, we have
$\card{g} \leq \capac{\bd(U)} < \frac{1}{4r} \bx(U)$.
Since $\card{f_1} + \card{g} \geq \frac{1}{2} \bx(U)$ (no flow path in $f_2$ originates in $U$), we have
that
$\card{g} \leq \frac{1}{4r} \cdot 2 \cdot  \left(\card{f_1} + \card{g}\right).$
Hence,
\begin{equation}\label{eq:instloss}
\card{g} \leq \frac{1}{2r} \cdot \left(1-\frac{1}{2r}\right)^{-1} \card{f_1} \leq \frac{1}{r} \card{f_1} \enspace .
\end{equation}
By putting up together~\eqref{eq:inst2}, \eqref{eq:inst1}, and~\eqref{eq:instloss}, we obtain that
\begin{align*}
\card{\sP_1} + \card{\sP_2} &\geq \frac{1}{144r^3} \left(1-\frac{1}{r}\right)^{\ell_1 + \ell_2} \left(\card{f_2} + \left(1-\frac{1}{r}\right)^{-1} \card{f_1}\right) \\
    & \geq \frac{1}{144r^3} \left(1-\frac{1}{r}\right)^{\ell_1 + \ell_2} \left(\card{f_2} + \card{f_1} + \card{g}\right) = \frac{1}{144r^3} \left(1-\frac{1}{r}\right)^{\ell_1+\ell_2} \card{f} \enspace .
\end{align*}
This concludes the proof of Theorem~\ref{thm:induction-main}. 
Since $\ell_1,\ell_2 \leq r$, while $(1-\frac{1}{r})^{2r} = \Omega(1)$, Theorem~\ref{thm:induction-main}
immediately implies the promised $\Oh(r^3)$-approximation algorithm.

\begin{remark}
We conclude with observing that the improved approximation ratio of $\Oh(r^3)$ directly translates to the more general setting
of $k$-sums of graph from some minor closed family $\mathcal{G}$, as discussed in~\cite{cns-tw}. That is, if we are able to $\alpha$-approximate \edp{} with congestion $\beta$
in graphs from $\mathcal{G}$, we can have $\Oh(\alpha r^5)$-approximation algorithm with congestion $(\beta+3)$ in graphs admitting a tree decomposition of maximum adhesion size at most~$r$,
and the torso of every bag being from the class $\mathcal{G}$.

To see this, observe that the only place when our algorithm uses that the \emph{bags} are of bounded size (as opposed to \emph{adhesions}) is the base case, where all flow paths pass through the bag $\bags(t_0)$ of the root node $t_0$.
However, in this case we can proceed exactly as Chekuri \etal~\cite{cns-tw}: using the flow paths, move the terminals to $\bags(t_0)$, replace connected components of $G-\bags(t_0)$ with their $(r^2,2)$-sparsifiers, and apply the algorithm
for the class $\mathcal{G}$. In addition to the $\Oh(r^3)$ approximation factor of our algorithm,
the application of the algorithm for~$\mathcal{G}$ incurs an approximation ratio of $\alpha$ and congestion of $\beta$, the use of sparsifiers adds a factor of $r^2$ to the approximation ratio
and an additive constant $+1$ to the congestion, while the terminal move adds an additional amount of $2$ to the final congestion.
\end{remark}


\section{Algorithm for \ndp in Bounded Pathwidth Graphs}
\label{sec:ndp}

In this section we develop an $\Oh(r^3)$-approximation algorithm for \ndp in graphs of \emph{pathwidth} less than~$r$.
We follow the outline of the \edp algorithm from the previous section, with few essential changes.

Most importantly, we can no longer use Proposition~\ref{prop:small-cutset-routing}, as it refers to edge disjoint paths,
and the proof of its main ingredient by Chekuri \etal~\cite{CKS-sqrtn} relies on a clustering technique that stops to work for node disjoint paths.
We fix this issue by providing in Sect.~\ref{ssec:ndp-small-cutset}
a node-disjoint variant of Proposition~\ref{prop:small-cutset-routing},
using the more involved clustering approach of Chekuri \etal~\cite{cks-stoc05}.

Then, in Sect.~\ref{ssec:ndp} we revisit step-by-step the arguments for \edp, pointing out remaining differences.
We remark that the use of pathwidth instead of treewidth is only essential in the inductive step for the case $\ell_1 < \ell_2$:
if we follow the argument for \edp for bounded-treewidth graphs, the graphs $G(t_i)$ may not be node disjoint (but they are
edge disjoint), breaking the argument. Note that for bounded pathwidth graphs, there is only one such graph considered,
and the issue is nonexistent.

\subsection{Routing to a small adhesion in a node-disjoint setting}\label{ssec:ndp-small-cutset}

In this section we prove the following statement.

\begin{proposition}\label{prop:ndp-small-cutset}
	Let $(G, \Pairs)$ be an instance of \ndp and let $(f, \bx)$ be a feasible fractional solution for the instance. Suppose that there is also a second (feasible, i.e., respecting node capacities) flow that routes at least $x(v)/\alpha$ units of flow for each $v$ to some set $S \subseteq V$, where $\alpha \geq 1$. Then there is an integral routing of $\Omega(\card{f} / (\alpha \card{S}))$ pairs.
\end{proposition}

Without loss of generality, we may assume that the terminals of $\Pairs$ are pairwise distinct and of degree
and capacity one: we can always move a terminal
from a vertex $t$ to a newly-created degree-$1$ capacity-$1$ neighbour of $t$.

Let $g$ be the second flow mentioned in the statement. In what follows, we modify and simplify the flows $f$ and $g$ in a number of steps.
We denote by $f_1,f_2,\ldots$ and $g_1,g_2,\ldots$ flows after subsequent modification steps; for the flow $f_i$, by $\bx_i$ we
denote its marginals.

\mypar{Symmetrizing the flow $g$.}
In the first step, we construct flows $f_1$ and $g_1$ with
the following property: for every terminal pair $(s,t) \in \Pairs$, for every $v \in S$,
$g_1$ sends the same amount of flow from $s$ to~$v$ as from $t$ to $v$.
To obtain this goal, we first take the flow $g/3$, and then for every $(s,t) \in \Pairs$ redirect the flow originating at $s$ to first go
along the commodity for the pair $(s,t)$ in flow $f/(3\alpha)$ to the vertex $t$, and then go to $S$ in exactly the same
manner as the flow originating at~$t$ does. It is easy to see that $g_1$ consists of three feasible flows
scaled down by at least $1/3$, thus it is feasible. Finally, we set $f_1 \coloneqq f/3$, so that $g_1$ again sends
$\bx_1(v)/\alpha$ flow from every vertex $v$ to $S$. Note that $\card{f_1} = \card{f}/3$.

\mypar{Restricting to single vertex of $S$.}
To construct flows $f_2$ and $g_2$, pick a vertex $u \in S$ that receives the most flow in $g_1$.
Take $g_2$ to be the flow $g_1$, restricted only to flow paths ending in $u$. Then, restrict $f_1$ to obtain $f_2$ as follows:
for every terminal pair $(s,t) \in \Pairs$, reduce the amount of flow from~$s$ to $t$ to $\alpha$ times the total amount
of flow sent from $s$ to $u$ by $g_2$; note that, by the previous step, it is also equal $\alpha$ times the total
amount of flow sent from $t$ to $u$ by $g_2$. By this step, we maintain the invariant that $g_2$ sends $\bx_2(v)/\alpha$
flow from every $v \in V(G)$, and we have $\card{f_2} \geq \card{f_1}/|S| \geq \card{f}/(3|S|)$.

\mypar{Rounding to a half-integral flow.}
In the next step, we essentially repeat the integral rounding procedure by Chekuri \etal~\cite[Section 3.2]{CKS-sqrtn}.
We use the following operation as a basic step in the rounding.

\begin{lemma}[Theorem 2.1 of~\cite{CKS-sqrtn}]\label{lem:single-rounding}
Let $G$ be a directed graph with edge capacities.
Given a flow~$h$ in~$G$ that goes from set $X \subseteq V(G)$ to a single vertex $u \in V(G)$,
 such that for every $v \in X$ the amount of flow
originating in $v$ is $\bz(v)$, and a vertex $v_0 \in X$ such that $\bz(v_0)$ is not an integer,
one can in polynomial time compute a flow $h'$ in $G$, sending $\bz'(v)$ amount of flow from every $v \in X$ to~$u$,
such that $\card{h'} \geq \card{h}$, $\bz'(v) = \bz(v)$ for every $v \in X$ where $\bz(v)$ is an integer,
and $\bz'(v_0) = \lceil \bz(v_0) \rceil$.
\end{lemma}

Since a standard reduction reduces flows in undirected node-capacitated graphs to directed edge-capacitated ones%
\footnote{Replace every edge with two infinite-capacity arcs in both directions, and then split every vertex into two vertices,
 connected by an edge of capacity equal to the capacity of the vertex, with all in-edges connected to the first copy,
 and all out-edges connected to the second copy.}, Lemma~\ref{lem:single-rounding} applies also to undirected graphs
with integral node capacities.

Split $g_2$ into two flows $h_s$ and $h_t$: for every terminal pair $(s,t) \in \Pairs$, we put the flow originating in~$s$
into~$h_s$, and the flow originating in $t$ into $h_t$. We perform a sequence of modifications to the flows~$h_s$ and $h_t$,
maintaining the invariant that the same amount of flow originates in $s$ in $h_s$ as in $t$ in $h_t$ for every $(s,t) \in \Pairs$.
Along the process, both $h_s$ and $h_t$ are feasible flows, but $h_s + h_t$ may not be.

In a single step, we pick a terminal pair $(s,t) \in \Pairs$ such that the amount of flow in $h_s$ originating in $s$ is not integral
(and stop if no such pair exists).
We apply Lemma~\ref{lem:single-rounding} separately to $s$ in $h_s$ and to $t$ in $h_t$, obtaining flows $h_s'$ and $h_t'$.
Finally, if for some terminal pair $(s',t')$, the amount of flow originating in $s'$ in~$h_s'$ and in $t'$ in $h_t'$ differ,
we restrict one of the flows so that both route the same amount of flow (being the minimum of the flows routed by $h_s'$ from $s'$ and
by $h_t'$ from $t'$).

Since the rounding algorithm of Lemma~\ref{lem:single-rounding} never modifies a source that already has an integral flow,
this procedure stops after at most $|\Pairs|$ steps. Furthermore, if in one step the flow from $s$ has been increased from~$z$ to $\lceil z \rceil$, the total loss of flow to other pairs is $2(\lceil z \rceil - z)$. Therefore,
if $h_s^\circ$ and~$h_t^\circ$ are the final integral flows, we have $\card{h_s^\circ} + \card{h_t^\circ} \geq (\card{h_s} + \card{h_t})/2
= \card{g_2}/2 = \card{f_2}/\alpha \geq \card{f}/(3\alpha\card{S})$. 
Finally, we define $g_3 := (h_s^\circ + h_t^\circ)/2$; note that $g_3$ is a feasible flow, since both $h_s^\circ$ and $h_t^\circ$ are.

\mypar{Clustering a node-flow-linked set.}
Note that for every $(s,t) \in \Pairs$, the flow~$g_3$ routes either 0 or $1/2$ flow from both $s$ and $t$
to $u$. Let $\Pairs'$ be the set of pairs for which the flow is $1/2$, and let~$X'$ be the set of terminals
in $\Pairs'$. Note that $\card{\Pairs'} = \card{g_3}/2 \geq \card{f}/(6\alpha\card{S})$.

Observe that the set $X'$ is $\frac{1}{4}$-node-flow-linked: using $g_3$, we can find a multicommodity
flow that for every $(a,b) \in X' \times X'$ routes $\frac{1}{4|X|}$ amount of flow from~$a$ to $b$, by routing it
along $g_3$ to $u$ and along reversed $g_3$ from $u$.
This allows us to apply the following clustering result.

\begin{lemma}[Lemma 2.7 of~\cite{cks-stoc05}]\label{lem:node-clustering}
If $X$ is $\alpha$-node-flow-linked in a graph $G$ with unit node capacities, then for any $h \geq 2$ there
exists a forest $F$ in $G$ of maximum degree $\Oh(\frac{1}{\alpha} \log h)$
such that every tree in $F$ spans at least~$h$ nodes from $X$.
\end{lemma}

Since we can assume that no capacity in $G$ exceeds $|\Pairs|$, we can replace every vertex $v$ of capacity $\capac{v}$
with its $\capac{v}$ copies. To such unweighted graph~$G'$ we
apply Lemma~\ref{lem:node-clustering} for $X'$, $\alpha=1/4$ and $h = 3$, obtaining a forest $F'$;
recall that the terminals $X'$ are of capacity $1$, thus they are kept unmodified in $G'$.
By standard argument we split the forest $F'$ into node-disjoint trees $T_1',T_2',\ldots,T_p'$, such that
every tree $T_i'$ contains at least three, and at most $d = \Oh(1)$ terminals of~$X'$. 
By projecting the trees $T_i'$ back onto $G$, we obtain a sequence of trees $T_1,T_2,\ldots,T_p$,
such that every vertex $v \in V(G)$ is present in at most $\capac{v}$ trees~$T_i$.
Furthermore, since terminals are of capacity one, every terminal belongs to at most one tree,
and every tree $T_i$ contains at least three and at most $d$ terminals.

In a greedy fashion, we chose a set $\Pairs'' \subseteq \Pairs'$ of size at least $\card{\Pairs'}/d^2$, such that
for every tree~$T_i$, at most one terminal pair of $\Pairs''$ has at least one terminal in $T_i$.
A pair $(s,t) \in \Pairs''$ is \emph{local} if both $s$ and $t$ lie in the same tree $T_i$, and \emph{distant} otherwise.
If at least half of the pairs of $\Pairs''$ are local, we can route them along trees $T_i$,
obtaining a desired routing of size at least $\card{\Pairs''}/2 \geq \card{\Pairs'}/(2d^2) = \Omega(\card{f}/(\alpha\card{S}))$
and terminate the algorithm.
Otherwise, we obtain a flow~$g_4$ as follows: for every terminal $t$ in a distant pair in $\Pairs''$, we take the tree $T_i$ it lies
on, route~$3/5$ amount of flow along $T_i$ 
equidistributed to three arbitrarily chosen terminals $t^1,t^2,t^3$ on~$T_i$ from~$\Pairs'$ (i.e., every terminal~$t^j$ receives $1/5$ amount of flow), and then route the flow
along the flow $\frac{2}{5} g_3$ to~$u$. Since every tree $T_i$ routes $3/5$ amount of flow, and $g_3$ is a feasible flow,
the flow~$g_4$ is a feasible flow that routes $3/5$ amount of flow from every terminal of~$\Pairs''$ to $u$.
Furthermore, since at least half terminal pairs in $\Pairs''$ is distant, 
  we have $\card{g_4} \geq \frac{1}{2} \cdot 2\card{\Pairs''} = \Omega(\card{f}/(\alpha\card{S}))$.

\mypar{Final rounding of the flow.}
Let $X''$ be the set of all terminals of $\Pairs''$.
Since the flow $g_4$ routes \emph{more than $1/2$} amount of flow for every terminal in $X''$, we 
can conclude with simple rounding the flow $g_4$ in the same manner as it is done by Chekuri \etal~\cite[Section 3]{CKS-sqrtn}.
Construct an auxiliary graph $G'$ by adding a super-source $s^\ast$ of infinite capacity, adjacent to all terminals
of~$\Pairs''$. Extend~$g_4$ in the natural manner, by routing every flow path first from $s^\ast$
to an appropriate terminal.
The extended flow $g_4$ is now a single souce single sink flow from $s^\ast$ to~$u$ in a graph with integer capacities,
thus there exists an integral flow $g_5$ of no smaller size:
$$\card{g_5} \geq \card{g_4} = \frac{3}{5} \card{X''} = \frac{6}{5} \card{\Pairs''}.$$
Hence, for at least $1/5$ of the pairs $(s,t) \in \Pairs''$, the flow $g_5$ routes a single unit of flow
both from~$s$ and from $t$ to $u$. By combining these paths into a single path from $s$ to $t$,
we obtain an integral routing of size at least $\frac{1}{5}\card{\Pairs''} = \Omega(\card{f}/(\alpha \card{S}))$.
This finishes the proof of Proposition~\ref{prop:ndp-small-cutset}.

\subsection{Details of the algorithm}\label{ssec:ndp}

Equipped with Proposition~\ref{prop:ndp-small-cutset}, we can now proceed to the description of the approximation algorithm.
Assume we are given an \ndp{} instance $(G,\Pairs)$ and a path decomposition $(\sT,\bags)$ of $G$ of width less than~$r$;
recall that $\sT$ rooted in one of its endpoints. Let $(f,\bx)$ be a fractional solution to the multicommodity flow
relaxation for \ndp, as in Section~\ref{sec:prelim}.

The definitions of safe and good node are analogous, and we follow the same induction scheme.
\begin{definition}[Safe node]
A node $t \in V(\sT)$ is \emph{safe} with respect to $(f, \bx)$
if there is a second multicommodity flow $g$ in $G(t)$ such that $g$ satisfies the node capacities of $G(t)$ and,
for each vertex $z \in \subtree(t)$, $g$ routes at least ${1 \over 4r} \cdot \bx(z)$ units of flow from $z$ to the adhesion $\adh(t)$.
The node $t$ is \emph{unsafe} if it is not safe.
\end{definition}

\begin{definition}[Good node]
A node $t \in V(\sT)$ is \emph{good} with respect to $(f, \bx)$ if every flow path in the support of $f$ that
has an endpoint in $\subtree(t)$ also intersects~$\adh(t)$;
in other words, there does not exist a flow path that is completely contained in~$G[\below(t)]$.
A node is \emph{bad} if it is not good.
\end{definition}

\begin{theorem} \label{thm:ndp-induction-main}
Let $(G, \Pairs)$ be an instance of \ndp and let $(f, \bx)$ be a fractional solution for the instance, where $f$ is a feasible multicommodity flow in $G$ for the pairs $\Pairs$ with marginals $\bx$. Let $(\sT,\bags)$ be a path decomposition for $G$ of width less than $r$.
Let $\ell_1$ be the maximum size of a parent adhesion of an unsafe node, and let $\ell_2$ be the maximum size of a
parent adhesion of a bad node.
There is a constant $c$ and a polynomial time algorithm that constructs an integral routing of size at least
$${1 \over c r^3} \cdot \left(1 - {1 \over r} \right)^{\ell_1 + \ell_2} \cdot \card{f} \enspace .$$
\end{theorem}

Again as in the case of \edp, we can assume that the considered graph $G$ is connected and that no bag is empty,
and thus the only empty adhesion is the parent adhesion of the root.

\mypar{Base case.}
In the base case $\ell_1 = \ell_2 = 0$ nothing changes as compared to \edp: all flow paths pass through the root bag,
and Proposition~\ref{prop:ndp-small-cutset} allows us to route integrally $\Omega(\card{f}/r)$ paths.

\mypar{Inductive step when $0 \leq \ell_1 < \ell_2$.}
Since we are considering now a path decomposition (as opposed to tree decomposition in the previous section),
there exists a single topmost bad node~$t^\circ$ with parent adhesion of size~$\ell_2$.
Define~$f^\mathrm{inside}$ to be the subflow of $f$ consisting of all flow paths completely contained in~$G[\below(t^\circ)]$.
Since $\ell_1 < \ell_2$, the node $t^\circ$ is safe, and the flow witnessing it
together with Proposition~\ref{prop:ndp-small-cutset} allows us to integrally route $\Omega(\card{f^\mathrm{inside}}/r^2)$ 
terminal pairs. If $\card{f^\mathrm{inside}} > \card{f}/r$, then we are done. 
Otherwise, we drop the flow $f^\mathrm{inside}$ from $f$, making $t^\circ$ and all its descendants good (thus decreasing
$\ell_2$ in the constructed instance), while losing only $1/r$ fraction of the flow $f$, and pass the instance
to an inductive step.

\mypar{Inductive step when $0 < \ell_1 = \ell_2$.}
Here again we take $t^\circ$ to be the lowest node of $\sT$ that is unsafe and has parent adhesion of size $\ell_1$.
By the definition of an unsafe node and Menger's theorem, there exists a set $U \subseteq \below(t^\circ)$
such that $\capac{N(U)} < \frac{1}{4r} \bx(U)$.
Using the same argument as in the proof of Lemma~\ref{lem:belowU}, we can ensure property~\ref{belowU:adh}, that is
that if $U$ contains an adhesion~$\adh(t)$, it contains as well the entire set $\subtree(t)$.

As in the case of \edp, we split into instances $\instance_1$ and $\instance_2$ by taking $G_1 = G[U]$
and $G_2 = G-N[U]$, with inherited tree decompositions from $(\sT,\bags)$.
Since all nodes with parent adhesions of size larger than $\ell_1= \ell_2$ are good, there are also good in instances $\instance_i$
(i.e., Lemma~\ref{lem:remains-good} holds here as well)
and we can again apply the inductive step to every connected component of
the instance $\instance_2$ with the same values of $\ell_1$ and $\ell_2$, obtaining a routing $\sP_2$
of size as in~\eqref{eq:inst2} (with $144$ replaced by a constant $c$).

We analyse the instance $\instance_1$, without breaking it first into connected components.
That is, we argue that in $\instance_1$ the value of $\ell_1$ dropped, that is,
all nodes $t$ satisfying $|\adh(t)| = |\adh_1(t)| = \ell_1$ are safe;
note that they will remain safe once we consider every connected component separatedly.
However, this fact follows from property~\ref{belowU:adh} of the set $U$ (Lemma~\ref{lem:belowU}):
if for some node $t$ we have $|\adh(t)| = |\adh_1(t)|$, it follows that $\adh(t) \subseteq U$ hence $\subtree(t) \subseteq U$
and the notion of safeness for $t$ is the same in $\instance_1$ and in the original instance.
However, $\adh(t) \subseteq U \subseteq \below(t^\circ)$
implies $t \preceq t^\circ$ and $t \neq t^\circ$, hence $t$ is safe in the original instance. 

Consequently, an application of inductive step for every connected component of $\instance_1$ uses strictly smaller value of $\ell_1$,
and we obtain an integral routing $\sP_1$ in~$\instance_1$ of size as in~\eqref{eq:inst1} (again with $144$ replaced by a constant $c$).
The remainder of the analysis from the previous section does not change, concluding the proof of Theorem~\ref{thm:ndp-induction-main}.\hfill$\Box$

\section{Well-linked decomposition}
\label{sec:wl-decomp}

In this section, we show that the argument given in Section~\ref{sec:edp} gives the following well-linked decomposition for edge-capacitated graphs. We first give some preliminary definitions.

\mypar{Multi-commodity flows.}
We represent a multi-commodity flow instance as a demand vector~$\vd$ that assigns a demand $d(u, v) \in \mathbb{R}_+$ to each ordered pair $(u, v)$ of vertices of $G$. A \emph{product} multi-commodity flow instance satisfies $d(u, v) = w(u) w(v)$ for each pair $(u, v)$, where $w: V \rightarrow \mathbb{R}_+$ is a weight functions on the vertices of $G$. We say that $\vd$ is \emph{routable} if there is a feasible multi-commodity flow in $G$ that routes $d(u, v)$ units of flow from $u$ to $v$ for each pair $(u, v)$.

We recall the following two quantities associated with a multi-commodity flow instance: the maximum concurrent flow and the sparsest cut. The \emph{maximum concurrent flow} is the maximum value $\lambda \geq 0$ such that $\lambda \vd$ is routable. The \emph{sparsity} of a cut $S \subseteq V$ is the ratio $\card{\delta(S)} / \demand_{\vd}(S)$, where $\demand_{\vd}(S) = \sum_{u \in S, v \in V \setminus S} d(u, v)$ is the total demand separated by $S$. A \emph{sparsest cut} is a cut with minimum sparsity. The minimum sparsity of any cut is an upper bound on the maximum concurrent flow, and the former could be strictly larger. The \emph{flow-cut gap} is the worst-case ratio between the sparsest cut and the maximum concurrent flow. The flow-cut gap is $\Oh(\log k)$ in general graphs, where $k$ is the number of commodities (that is, the number of non-zero demands $d(u, v)$) \cite{LeightonR99}. For product multi-commodity flows, the flow-cut gap is $\Oh(\log r)$ for graphs of treewidth $r$ \cite{chekuri2009note}. Moreover, for both of these results, there are polynomial time algorithms for computing a cut with sparsity at most $\alpha(G) \cdot \lambda$, where $\alpha(G)$ is the flow-cut gap value and $\lambda$ is the maximum concurrent flow for $\vd$ \cite{LeightonR99,chekuri2009note}.

\mypar{Well-linked sets.}
Following \cite{cks-stoc05}, we work with two notions of well-linked sets, cut well-linked sets and flow well-linked sets. For convenience, we work with graphs with unit edge capacities; it is straightforward to extend the argument to arbitrary edge capacities.

Let $\pi: X \rightarrow \mathbb{R}_+$ be a weight function on a set $X \subseteq V(G)$ of vertices. The set $X$ is \emph{$\pi$-flow-well-linked} in $G$ if there is a feasible multi-commodity flow that simultaneously routes $d(u, v) \coloneqq \pi(u) \pi(v) / \pi(X)$ units of flow from $u$ to~$v$ for every pair $(u, v)$ of nodes in $X$. The set $X$ is \emph{$\pi$-cut-well-linked} in $G$ if $\card{\delta(S)} \geq \min\set{\pi(S \cap X), \pi((V \setminus S) \cap X)}$.

Chekuri \etal \cite{cks-stoc05} gave the following well-linked decomposition theorem for general graphs that has found many applications. We only state the theorem for flow-well-linked instances, there is an analogous result for cut-well-linked instances.

\begin{theorem}[\cite{cks-stoc05}]
	Let $\opt$ be the value of a solution to \LP for a given instance $(G,
	\Pairs)$ of \edp on a general graph $G$. Let $\alpha = \alpha(G) \geq 1$ be an upper
	bound on the worst case flow-cut gap for product multi-commodity
	flows in~$G$. There is a partition of $G$ into node-disjoint
	induced subgraphs $G_1, G_2, \dots, G_q$ and weight
	functions $\pi_i: V(G_i) \rightarrow \mathbb{R}_+$ with the
	following properties. Let~$\Pairs_i$ be the induced pairs of
	$\Pairs$ in $G_i$ and let $X_i$ be the endpoints of the pairs in
	$\Pairs_i$. We have
	\begin{enumerate}
		\item[(a)] $\pi_i(u) = \pi_i(v)$ for each pair $uv \in \Pairs_i$.
		\item[(b)] $X_i$ is $\pi_i$-flow-well-linked in $G_i$.
		\item[(c)] $\sum_{i = 1}^{q} \pi_i(X_i) = \Omega(\opt/
		(\alpha \log{\opt})) = \Omega(\opt / \log^2 k)$.
	\end{enumerate}
	Moreover, such a partition is computable in polynomial time if
	there is a polynomial time algorithm for computing a node
	separator with sparsity at most $\alpha(G)$ times the maximum
	concurrent flow.
\end{theorem}

In the remainder of this section, we show that, for graphs with treewidth at most $r$, the argument in Sect.~\ref{sec:edp} implies a well-linked decomposition that loses a factor of $\Oh(r^3)$ instead of $\Oh(\log^2k)$.

\begin{theorem} \label{thm:wl-decomp}
	Let $\opt$ be the value of a solution to \LP for a given instance $(G, \Pairs)$ of \edp on a graph $G$ of treewidth at most $r$. There is a partition of $G$ into node-disjoint induced subgraphs $G_1, G_2, \dots, G_q$ and weight functions $\pi_i: V(G_i) \rightarrow \mathbb{R}_+$ with the following properties. Let~$\Pairs_i$ be the induced pairs of $\Pairs$ in $G_i$ and let $X_i$ be the endpoints of the pairs in $\Pairs_i$. We have
	\begin{enumerate}
		\item[(a)] $\pi_i(u) = \pi_i(v)$ for each pair $uv \in \Pairs_i$.
		\item[(b)] There is a feasible multi-commodity flow $g_i$ in $G_i$ from $X_i$ to a single vertex $z$ that routes $\pi_i(v)$ units of flow from $v$ to $z$ for each vertex $v \in X_i$. Thus $X_i$ is $\pi_i$-flow-well-linked in $G_i$.
		\item[(c)] $\sum_{i = 1}^{q} \pi_i(X_i) = \Omega(\opt/ r^3)$.
	\end{enumerate}
	Moreover, such a partition is computable in polynomial time if one is given a tree decomposition of $G$ of width at most $r$.
\end{theorem}

In the remainder of this section, we show how to modify the argument in Section~\ref{sec:edp} to prove Theorem~\ref{thm:wl-decomp}. We use the definitions and notation introduced in Section~\ref{sec:edp}.

As before, our starting point is a fractional solution $(f, \bx)$ to \LP for the instance $(G, \Pairs)$. We prove the following theorem by induction on the maximum size of a parent adhesion
that is bad or unsafe.

\begin{theorem} \label{thm:wl-induction-main}
	Let $(G, \Pairs)$ be an instance of \edp and let $(f, \bx)$ be a fractional solution for the instance, where $f$ is a feasible multicommodity flow in $G$ for the pairs $\Pairs$ with marginals $\bx$. Let $(\sT,\bags)$ be a tree decomposition for $G$ of width less than $r$. Let $\ell_1$ be the maximum size of a parent adhesion of an unsafe node, and let $\ell_2$ be the maximum size of a parent adhesion of a bad node. There is a polynomial time algorithm that constructs a partition of $G$ into node-disjoint induced subgraphs $G_1, G_2, \dots, G_q$ and weight functions $\pi_i: V(G_i) \rightarrow \mathbb{R}_+$ with the following properties. Let~$\Pairs_i$ be the induced pairs of $\Pairs$ in $G_i$ and let $X_i$ be the endpoints of the pairs in $\Pairs_i$. We have
	\begin{enumerate}
		\item[(a)] $\pi_i(u) = \pi_i(v)$ for each pair $uv \in \Pairs_i$.
		\item[(b)] There is a feasible multi-commodity flow $g_i$ in $G_i$ from $X_i$ to a single vertex $z$ that routes $\pi_i(v)$ units of flow from $v$ to $z$ for each vertex $v \in X_i$. Thus $X_i$ is $\pi_i$-flow-well-linked in $G_i$.
		\item[(c)] $\sum_{i = 1}^{q} \pi_i(X_i) \geq {1 \over 12 r^3} \cdot \left(1 - {1 \over r} \right)^{\ell_1 + \ell_2} \cdot \card{f}$.
	\end{enumerate}
\end{theorem}

In addition to the definitions of a good and safe node from Section~\ref{sec:edp}, it is convenient to have the following definition.

\begin{definition}[Nice flow]
  A multi-commodity flow $f$ is \emph{nice} for a node $t \in V(\sT)$ if there exists a single vertex $z \in \adh(t)$ such that each flow path $P$ in the support of $f$ ends at $z$ and $P - z$ is contained in $G[\below(t)]$.
\end{definition}

The proof of the following lemma follows from a result by Chekuri et al.~\cite[Proposition 3.4]{cns-tw}, and we include it for completeness.

\begin{lemma}
\label{lem:nice-flow}
  Consider a node $t \in V(\sT)$. Let $\Pairs$ be a collection of pairs with both endpoints in~$\below(t)$ such that each vertex of $\below(t)$ appears in at most one pair. Let~$f$ be a multi-commodity flow in $G[\below(t)]$ with marginals~$\bx$ satisfying $x(u) = x(v)$ for each pair $(u, v) \in \Pairs$. Suppose that there is a second multi-commodity flow~$g$ in~$G(t)$ that routes at least $x(v) / c$ units of flow for each vertex $v \in V(\Pairs)$ to~$\adh(t)$, where $c \geq 1$. Then there is a multi-commodity flow $h$ from $V(\Pairs)$ to $z$ with the following properties: $h$ is nice for $t$; for each pair $(u, v) \in \Pairs$, $h$ routes the same amount of flow for $u$ and $v$; $\card{h} \geq  {1 \over 3c \card{\adh(t)}} \cdot \card{f} \geq {1 \over 3cr} \cdot \card{f}$. 
\end{lemma}
\begin{proof}
  Note that we may assume that each flow path in the support of $G(t)$ is contained in $G[\below(t)]$ except for its end vertex. Let $z \in \adh(t)$ be a vertex that receives the most $g$-flow, where the $g$-flow received by $z$ is the total amount of flow on paths ending at $z$. Let $g'$ be the sub-flow of $g$ consisting of only the flow paths ending at $z$, and let $\bx'$ be the marginals of $g'$.

  Note that $x'(v) \leq x(v) / c$ for every $v \in V(\Pairs)$ and $\card{g'} \geq {1 \over c \card{\adh(t)}} \sum_{v \in \below(t)} \card{f}$. However, there may be pairs $(u, v) \in \Pairs$ for which $x'(u) \neq x'(v)$. We use the flow $f$ to ensure the latter property as follows. Consider a pair $(u, v) \in \Pairs$ such that $x'(u) \neq x'(v)$, and suppose without loss of generality that $x'(u) < x'(v)$. We create $x'(v)$ units of flow from $u$ to $z$ as follows: since $x'(v) \leq x(v)/c$, we can route $x'(v)$ units of flow from $u$ to $v$ using the flow paths of $f$ and then route $x'(v)$ units of flow from $v$ to $z$ using the flow paths of $g'$. Therefore $f + 2 g'$ contains a congestion three sub-flow that routes $x'(u) = x'(v)$ units of flow for each pair $(u, v) \in \Pairs$. If we scale down this flow by a factor of $3$, we obtain a feasible flow $h$ with the desired properties.
\end{proof}

We preprocess the tree decomposition as in Section~\ref{sec:edp}. Thus we may assume that $G$ is connected and no bag is empty.

We prove Theorem~\ref{thm:induction-main} by induction on $\ell_1 + \ell_2 + \card{V(G)}$.

\mypar{Base case.}
In the base case, we assume that $\ell_1 = \ell_2 = 0$. Since every parent adhesion of a non-root node is non-empty,
that implies that the only bad node is the root $t_0$, that is, 
every flow path in $f$ passes through~$\bags(t_0)$, which is of size at most $r$.

We can obtain a flow $g$ from $\Pairs$ to $\bags(t_0)$ as follows: for each path $P$ in the support of $f$, let $P'$ be the smallest prefix of $P$ that ends at a vertex of $\bags(t_0)$, and let $g(P') = f(P)$ (since each path~$P$ in the support of~$f$ intersects $\bags(t_0)$, there exists such a prefix $P'$). Using the flows $f$ and $g$, we can construct a flow $h$ in $G$ from $V(\Pairs)$ to a single vertex $z \in \bags(t_0)$ with marginals $\pi$ such that $\pi(u) = \pi(v)$ for each pair $(u, v) \in \Pairs$ and $\card{h} \geq {1 \over 3r} \card{f}$. (This follows from the argument used in the proof of Lemma~\ref{lem:nice-flow}.). Thus $(G, \pi)$ is the desired decomposition.

In the inductive step, we consider two cases, depending on whether $0 \leq \ell_1 < \ell_2$ or $0 < \ell_1 = \ell_2$. The latter case can be handled in the same way as in Section~\ref{sec:edp}. The former case requires a different argument.

\mypar{Inductive step when $0 \leq \ell_1 < \ell_2$.}
Let $\{t_1,t_2,\ldots,t_p\}$ be the topmost bad nodes of $\sT$ with parent adhesions of size $\ell_2$, that is, 
it is a minimal set of such bad nodes such that for every bad node $t$ with parent adhesion of size $\ell_2$,
there exists $1 \leq i \leq p$ with $t \preceq t_i$.
For $1 \leq i \leq p$, let $f^{\mathrm{inside}}_i$ be the subflow of $f$ consisting of all paths that are completely contained in
$G[\alpha(t_i)]$. Furthermore, since $\ell_1 < \ell_2$, the node $t_i$ is safe; let $g_i$ be the corresponding flow, i.e., 
a flow that routes $\frac{1}{4r} x(v)$ from every $v \in \subtree(t_i)$ to $\adh(t_i)$ in $G(t_i)$. By applying Lemma~\ref{lem:nice-flow}, for each $i$ such that $1 \leq i \leq p$, there is a flow $h_i$ that is nice for $t_i$ and satisfies $\card{h_i} \geq {1 \over 12 r^2} \card{f^{\mathrm{inside}}_i}$. 

If $\sum_{i=1}^p \card{f^{\mathrm{inside}}_i} > \frac{1}{r} \card{f}$, we construct the desired decomposition as follows. Let $z_i \in \adh(t_i)$ be the endpoint of all the flow paths of $h_i$. Let $H_i = G[\alpha(t_i)] \cup \set{z_i}$. Note that $h_i$ is completely contained in~$H_i$ and the graphs $H_i$ are node-disjoint except for the vertices $z_i$. We group together the graphs with the same~$z_i$ vertex as follows. For each distinct vertex $z \in \set{z_1, \dots, z_p}$, let $H_z = \bigcup \set{H_i \colon 1 \leq i \leq p, z_i = z}$ and $h_z = \sum\set{h_i \colon 1 \leq i \leq p, z_i = z}$. Note that $h_z$ is a feasible flow in $H_z$ whose flow paths all end at $z$. The desired decomposition has a component $(H_z, \pi_z)$ for each distinct vertex $z$, where $\pi_z$ are the marginals of $h_z$. 

Therefore we may assume that $\sum_{i=1}^p \card{f^{\mathrm{inside}}_i} \leq \frac{1}{r} \card{f}$. In this case, we drop the flows $f^{\mathrm{inside}}_i$, that is, consider a flow
$f' \coloneqq f - \sum_{i=1}^p f^{\mathrm{inside}}_i$.
Clearly, $\card{f'} \geq (1-\frac{1}{r}) \card{f}$.
Furthermore, by the definition of $f^\mathrm{inside}_i$, every node $t_i$ is good with respect to $f'$.
Since deleting a flow path cannot turn a good node into a bad one nor a safe node into an unsafe one, and all descendants 
of a good node are also good, we infer that every unsafe node with respect to~$f'$ has parent adhesion of size at most~$\ell_1$,
while every bad node with respect to~$f'$ has parent adhesion of size \emph{less} than $\ell_2$. 
Consequently, by induction hypothesis we obtain a decomposition $(G_1, \pi_1), \dots, (G_q, \pi_q)$ satisfying
  $$\sum_{i = 1}^{q} \pi_i(X_i) \geq \frac{1}{12 r^3} \left(1-\frac{1}{r}\right)^{\ell_1 + \ell_2 - 1} \card{f'} \geq \frac{1}{12 r^3} \left(1-\frac{1}{r}\right)^{\ell_1 + \ell_2} \card{f}.$$

\mypar{Inductive step when $0 < \ell_1 = \ell_2$.}
This case follows from the argument given in Section~\ref{sec:edp}. We define two node-disjoint instances $\instance_1$ and $\instance_2$ exactly as before. We apply the induction hypothesis to each of the instaces and obtain a decomposition for each. Since the instances are node-disjoint, the union of the two decompositions  is the desired decomposition.

This concludes the proof of Theorem~\ref{thm:wl-induction-main}, which immediately implies Theorem~\ref{thm:wl-decomp}.
$~$\hfill$\Box$

\section{Hardness for \ndp in bounded treedepth graphs}
\label{sec:ndp-lb}

In this section we prove Theorem~\ref{thm:ndp-lb}.

We reduce from the $W[1]$-hard \textsc{Multicolored Clique} problem~\cite{FellowsEtAl2009}, where given a graph
$G$, an integer~$k$, and a partition $V = V^1 \uplus V^2 \uplus \ldots \uplus V^k$,
we are to check if there exists $k$-clique in $G$ with exactly one
vertex in every set $V^i$. By adding dummy vertices, we can assume that $|V^i| = n$ for every $1 \leq i \leq k$, and that $n, k \geq 2$.

\paragraph{Construction.}
Given an instance $(G,k,(V^i)_{i=1}^k)$ of \textsc{Multicolored Clique}, we aim at constructing an equivalent instance $(H,\termprs,\ell)$ of \ndp.

We start with a construction, for every set $V^i$, a gadget $W^i$ as follows.
First, for every $v \in V^i$ we construct an $(k-1)$-vertex path $X_v^i$ on vertices $x_{v,1}^i, x_{v,2}^i, \ldots,  x_{v,i-1}^i, x_{v,i+1}^i, \ldots, x_{v,k}^i$.
Second, we select an arbitrary vertex $u^i \in V_i$.
Third, for every $v \in V^i \setminus \{u^i\}$, we add a vertex $s^i_v$ adjacent to the first vertex of $X_v^i$ (i.e., $x_{v,1}^i$ and $x_{u^i,1}^i$ if $i > 1$ or $x_{v,2}^i$ and $x_{u^1,2}^i$ if $i=1$),
  a vertex $t^i_v$ adjacent to the last vertex of $X_v^i$ (i.e., $x_{v,k}^i$ and $x_{u^i,k}^i$ if $i < k$ or $x_{v,k-1}^i$ and $x_{u^i,k-1}^i$ if $i=k$), and make $(s^i_v,t^i_v)$ a terminal pair.
This concludes the description of the gadget $W^i$. By $\termprs_{st}$ we denote the set of terminal pairs constructed in this step.

To encode adjacencies in $G$, we proceed as follows. For every pair $1 \leq i < j \leq k$, we add a vertex $p_{i,j}$, adjacent
to all vertices $x_{v,j}^i$ for $v \in V_i$ and all vertices $x_{u,i}^j$ for $u \in V_j$. For every edge $vu \in E(G)$ with $v \in V_i$ and $u \in V_j$,
we add a terminal pair $(x_{v,j}^i, x_{u,i}^j)$. Let $\termprs_x$ be the set of terminal pairs constructed in this step; we have $\termprs = \termprs_{st} \cup \termprs_x$.

Finally, we set the required number of paths $\ell := k(n-1) + \binom{k}{2}$. This concludes the description of the instance $(H,\termprs,\ell)$; see Figure~\ref{fig:ndp-lb} for an illustration.

\begin{figure}[htb]
\begin{center}
\includegraphics{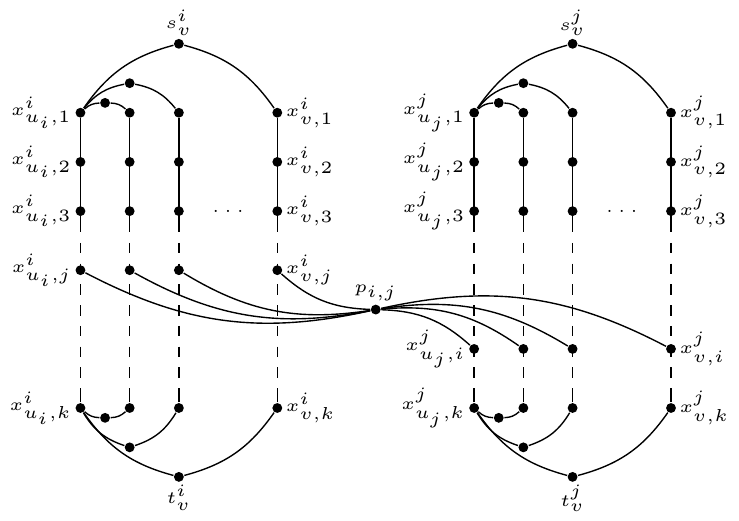}
\caption{Part of the construction of the graph $H$: gadgets $W^i$ and $W^j$, together with a connection via $p_{i,j}$.}
\label{fig:ndp-lb}
\end{center}
\end{figure}

\paragraph{From a clique to disjoint paths.}
Assume that the input \textsc{Multicolored Clique} instance is a ``yes''-instance, and let $\{v^i : 1 \leq i \leq k\}$ be a clique in $G$ with $v^i \in V^i$ for every $1 \leq i \leq k$.
We construct a family of~$\ell$ vertex-disjoint paths as follows. First, for every $1 \leq i \leq k$ and every $v \in V^i \setminus \{u^i\}$, we route a path from~$s^i_v$ to $t^i_v$
through the path $X_v^i$ if $v \neq v^i$, and through the path $X_{u^i}^i$ if $v = v^i$. Note that in this step we have created $k(n-1)$ vertex-disjoint paths
connecting terminal pairs, and in every gadget $W^i$ the only unused vertices are vertices on the path~$X_{v^i}^i$.
To construct the remaining $\binom{k}{2}$ paths, for every pair $1 \leq i < j \leq k$ we take the $3$-vertex path from $x_{v^i,j}^i$ to $x_{v^j,i}^j$ through $p_{i,j}$;
note that the assumption that $v^iv^j \in E(G)$ ensures that $(x_{v^i,j}^i, x_{v^j,i}^j)$ is indeed a terminal pair in $\termprs$.

\paragraph{From disjoint paths to a clique.}
In the other direction, let $\Pfam$ be a family of $\ell$ vertex-disjoint paths connecting terminal pairs in $H$. 
Let $\Pfam_{st} \subseteq \Pfam$ be the set of paths connecting terminal pairs from $\termprs_{st}$, and similarly define $\Pfam_x$.
First, observe that the set $P := \{p_{i,j} : 1 \leq i < j \leq k\}$ separates every terminal pair from~$\termprs_x$. Hence, every path from $\Pfam_x$ contains at least one vertex from $P$.
Since $|P| = \binom{k}{2}$, we have $|\termprs_x| \leq \binom{k}{2}$, and, consequently, $|\Pfam_{st}| \geq \ell - \binom{k}{2} = k(n-1) = |\termprs_{st}|$. 
We infer that $\Pfam_{st}$ routes all terminal pairs in $\termprs_{st}$ without using any vertex of $P$, while $\Pfam_x$ routes $\binom{k}{2}$ pairs from $\termprs_x$, and every path from $\Pfam_x$ contains exactly one vertex from $P$.

Since the paths in $\Pfam_{st}$ cannot use any vertex in $P$, every such path needs to be contained inside one gadget~$W^i$. Furthermore, observe that a shortest path between terminals $s_{v,a}^i$ and $t_{v,a}^i$ inside $W^i$ is either $X_{u^i}^i$ or~$X_v^i$,
      prolonged with the terminals at endpoints, and thus contains $k+1$ vertices.
Furthermore, a shortest path between two terminals in $\termprs_x$ contains three vertices. We infer that the total number of vertices on paths in $\Pfam$ is at least
$$|\Pfam_{st}| \cdot (k+1) + |\Pfam_x| \cdot 3 = k(n-1)(k+1) + 3\binom{k}{2} = k\left(n(k-1) + 2(n-1)\right) + \binom{k}{2} = |V(H)|.$$
We infer that every path in $\Pfam_{st}$ consists of $k+1$ vertices, and every path in $\Pfam_x$ consists of three vertices.
In particular, for every $1 \leq i \leq k$ and $v \in V^i \setminus \{u^i\}$, the path in $\Pfam_{st}$ that connects $s_v^i$ and $t_v^i$ goes either through~$X_v^i$ or~$X_{u^i}^i$.
Consequently, for every $1 \leq i \leq k$ there exists a vertex $v^i \in V^i$ such that the vertices of~$W^i$ that do not lie on any path from $\Pfam_{st}$ are exactly the vertices on the path $X_{v^i}^i$. 

We claim that $\{v^i : 1 \leq i \leq k\}$ is a clique in $G$. To this end, consider a pair $1 \leq i < j \leq k$.
Since $|\Pfam_x| = \binom{k}{2}$, there exists a path in $\Pfam_x$ that goes through~$p_{i,j}$. Moreover, this path has exactly three vertices. Since the only neighbours of $p_{i,j}$ that are not used
by paths from $\Pfam_{st}$ are $x_{v^i,j}^i$ and $x_{v^j,i}^j$, we infer that $(x_{v^i,j}^i, x_{v^j,i}^j) \in \termprs$ and, consequently, $v^iv^j \in E(G)$. This concludes the proof of the correctness of the construction.

\paragraph{Treedepth bound.}
We are left with a proof that $H$ has bounded treedepth. To this end, we use the alternative definition of treedepth~\cite{td-bound}: treedepth of an empty graph is $0$, while for
any graph~$G$ on at least one vertex we have that
$$\treedepth(G) =
  \begin{cases}
    1 + \min\{\treedepth(G-v) : v \in V(G)\} & \textrm{if }G\textrm{ is connected} \\
    \max\{\treedepth(C) : C \textrm{ connected component of } G\} & \textrm{otherwise.} 
  \end{cases}$$
First, observe that $H-P$ contains $k$ connected components, being the gadgets~$W^i$. Second, observe that the deletion of the endpoints of the path $X_{u^i}^i$ from the gadget $W^i$ breaks $W^i$ into connected components being paths on at most $k+1$ vertices.
Consequently,
$$\treedepth(H) \leq |P| + 2 + k+1 = \Oh(k^2).$$
This finishes the proof of Theorem~\ref{thm:ndp-lb}. \hfill$\Box$

\bibliographystyle{abbrv}
\bibliography{routing-treewidth}

\newpage
\appendix

\end{document}